\documentclass[journal]{IEEEtran}
\usepackage{array}
\usepackage{mathptmx}      
\usepackage{amsmath}       
\usepackage{amsthm}
\usepackage{pbox}
\usepackage{verbatim}
\usepackage{color}
\usepackage{amssymb}
\usepackage{enumerate}
\usepackage[utf8]{inputenc}
\usepackage[english]{babel}
\usepackage{cite}
\usepackage{todonotes}

\usepackage{txfonts}

\usepackage{graphicx}
\graphicspath{{./img/}}
\DeclareGraphicsExtensions{.pdf}

\newcommand{\myincludegraphics}[2]{
\begin{figure}
  \centering
  \includegraphics[scale=0.50]{img_#1}
  \caption{#2}
  \label{fig:#1}
\end{figure}
}

\newcommand{\myfigref}[1]{Figure~\ref{fig:#1}}

\newtheorem{definition}{Definition}
\newtheorem{theorem}{Theorem}
\newtheorem{corollary}{Corollary}[theorem]

\newtheorem{proposition}{Proposition}
\newtheorem{assumption}{Assumption}

\usepackage{soul}

\begin{document}

\title{Cyber Risk Analysis of Combined Data Attacks Against Power System State Estimation}

\author{Kaikai~Pan,~\IEEEmembership{Student~Member,~IEEE,}
        Andr\'e~Teixeira,~\IEEEmembership{Member,~IEEE,}
        Milos~Cvetkovic,~\IEEEmembership{Member,~IEEE,}
        and~Peter~Palensky,~\IEEEmembership{Senior~Member,~IEEE}
}

\maketitle

\begin{abstract}
Understanding smart grid cyber attacks is key for developing appropriate protection and recovery measures. Advanced attacks pursue maximized impact at minimized costs and detectability.
This paper conducts risk analysis of combined data integrity and availability attacks against the power system state estimation. We compare the combined attacks with pure integrity attacks - false data injection (FDI) attacks. A security index for vulnerability assessment to these two kinds of attacks is proposed and formulated as a mixed integer linear programming problem. We show that such combined attacks can succeed with fewer resources 
than FDI attacks. The combined attacks with limited knowledge of the system model also expose advantages in keeping stealth against the bad data detection. Finally, the risk of combined attacks to reliable system operation is evaluated using the results from vulnerability assessment and attack impact analysis. The findings in this paper are validated and supported by a detailed case study.
\end{abstract}


\IEEEpeerreviewmaketitle

\section{Introduction}

\IEEEPARstart{T}{he} increasingly digitized power system offers more data, details, and controls in a real-time fashion than its non-networked predecessors. One of the benefiting applications of this development is State Estimation (SE): Remote Terminal Units (RTUs) provide measurement data via Information and Communication Technology (ICT) infrastructure such as Supervisory Control and Data Acquisition (SCADA) system. The SE provides the operator with an estimate of the state of the electric power system. This state information is then used and processed by the energy management system (EMS) for optimal power flow (OPF), contingency analysis (CA), and automatic generation control (AGC). Security of supply depends on the EMS, which in turn depends on a reliable SE.

As discussed in \cite{Giani2009}, 
the SCADA system is vulnerable to a large number of security threats. A class of integrity data attack, known as false data injection (FDI) attack, has been studied with 
considerable attention. With modifying the measurement data, this attack can pass the Bad Data Detection (BDD) within SE to keep stealth \cite{LiuNingReiter2009}, by tampering of RTUs, the communication links to the control center, or even the databases and IT software in the control center. However, such FDI attack needs intensive attack resources such as the knowledge of the system model and the capability to corrupt the integrity on a set of measurements. 
Denial-of-service (DoS) attacks \cite{Wang2013} \cite{Deka2015}, a type of availability attack, are much ``cheaper'' to achieve, especially if RTUs communicate via insecure communication channels.
In this paper, we focus on combined attacks 
where the SE is corrupted by both integrity attacks and availability attacks simultaneously.
We compare combined attacks and FDI attacks under different levels of adversarial knowledge and resources.

\subsection{State of the Art}

Research in the literature has focused on FDI attacks from many aspects of risk assessment \cite{Ronald2012}, e.g., vulnerability analysis, attack impact assessment and mitigation schemes development. As first shown in \cite{LiuNingReiter2009}, a class of FDI attack, so-called \emph{stealth} attack, can perturb the state estimate without triggering alarms in BDD within SE. Vulnerability of SE to \emph{stealth FDI attacks} is usually quantified by computing attack resources needed by the attacker to alter specific measurements and keep stealth against the BDD \cite{Hug2012, 
Sandberg2010, TeixeiraSou2015}. 

Since state estimates are inputs of many application specific tools in EMS, the corrupted 
estimates can infect further control actions. The estimate errors due to FDI attacks were analyzed in \cite{TeixeiraDanSandbergEtAl2011} and \cite{KosutJiaThomasEtAl2011}. The results illustrate that the errors could be significant even with a small number of measurements being compromised. The work in \cite{XieMoSinopoli2011} and \cite{Jia2014} studied the potential economic impact of FDI attacks against SE by observing the nodal price of market operation. The attacker could obtain economic gain or cause operating costs in the market. Recent work in \cite{Liang2016} studied the physical impact of FDI attacks with the attacker's goal to cause a line overflow. 

In order to defend against stealth FDI attacks, mitigation schemes have been proposed to improve the bad data detection algorithm or safeguard certain measurements from adversarial data injection. Sequential detection (or quickest detection) of FDI attacks was designed mainly based on well-known Cumulative Sum (CUSUM) algorithm in \cite{Li2015}. In reference \cite{Ashok2016},  
detection methods that leverage synchrophasor data and other forecast information were presented. 
The network layer and application layer mitigation schemes, such as multi-path routing and data authentication and protection, are proved to be effective to decrease the vulnerability \cite{VukovicSouDanEtAl2012} \cite{Pan2016}. 

It is worth noting that the majority of research has focused on stealth FDI attacks from a specific aspect of vulnerability or impact assessment. The work in \cite{Deka2015} first considered adding a class of availability attack, so-called jamming attack, to the attack scenarios against SE. Our recent paper \cite{Pan2016} studied the \emph{stealth combined attacks} with different measurement routing topologies, concluding that such attacks may need less attack resources than FDI attacks. The work above assumed that the adversary has full knowledge of the system model, yielding perfect stealth attacks. However, the data of the system model is usually protected well and hard to be accessed by the adversary. In reality, the attacks are always executed with limited adversarial knowledge and have the possibility to be detected by the BDD under limited knowledge conditions. Thus for the vulnerability analysis, not only the attack resources needed by the attacker should be considered but also the detection probability of attacks needs to be computed. In addition, vulnerability and impact of attacks can be combined together in the notion of $risk$.  
In \cite{Sridhar2012}, a high-level risk assessment methodology for power system applications including SE was presented. However, risk analysis methods and tools combining vulnerability and impact assessment for data attacks are needed to implement risk assessment methodologies. 

In this paper, we extend our prior work reported in \cite{Pan2016} to formulate combined attacks with limited adversarial knowledge of the system model and conduct the risk analysis. 
In order to assess the risk, we first analyze vulnerability of SE with respect to attack resources needed by the adversary and calculate the detection probability of combined attacks. Next, we propose attack impact metric for evaluating attack impact on load estimate. Combining the results from vulnerability and impact assessment, we present the $risk$ which combined attacks bring to reliable system operation. We compare the vulnerability, impact and risk with those of FDI attacks. The simulation results show that combined attacks yield higher risk in majority of considered cases.

\subsection{Contributions and Outline}  

As far as we know, our work is the first one to conduct risk analysis of combined attacks with limited adversarial knowledge. Our contributions are listed as follows:

\begin{itemize}
\item[1)] The first part of vulnerability analysis is presented through the notion of security index \cite{Sandberg2010}, which corresponds to the minimum attack
resources needed by the attacker to compromise the measurements while keeping stealth. The power system is more vulnerable to attacks with smaller security index since such attacks can be executed with less resources. 
We show that, when availability attack and integrity attack have the same cost, the security indexes of combined attacks and FDI attacks coincide. 
 
\item[2)] Our second contribution is to address the detection probability problem of combined attacks with limited adversarial knowledge. Here we relax the full knowledge assumption which is commonly used in the literature.  
We show that the optimal combined attack with limited adversarial knowledge can still keep stealth under certain conditions. 
The empirical results also indicate that combined attacks have lower detection probability. 

\item[3)] We propose risk metric to quantify the risk of combined attacks with limited adversarial knowledge. For the attacks with the same security index, the risk metric is computed by multiplying 1) the probability of the attack not to be detected, with 2) the attack impact on load estimate. We particularly consider the attack impact on load estimate because the load estimates are inputs of other applications that compute optimal control actions in EMS. 
Based on the analysis of risk metrics of combined attacks and FDI attacks, we show that power system operations face higher risk under combined attacks.

\end{itemize}

The outline of the paper is as follows. Section II gives an introduction of SE 
and stealth FDI attacks mechanism. Section III extends the attack scenario to combined attacks and proposes security index with computational method for vulnerability analysis. In Section IV, the detectability of combined attacks with limited adversarial knowledge is discussed.
The risk metric is proposed to measure the risk of attacks in Section V with the analysis of the vulnerability and attack impact. 
Section VI presents empirical results from a power system use case. In section VII we conclude the paper.  

\subsection{Notation}

For an $m \times n$ matrix $H \in \mathbb{R}^{m \times n}$, we denote the $i$-th row of $H$ by $H(i,:)$. For a vector of $m$ values $a \in \mathbb{R}^{m}$, $a(i)$ is the $i$-th entry of $a$. By $diag(a)$, we denote an $m \times m$ diagonal matrix with the elements of vector $a$ on the main diagonal.

\section{Power System Model and Data Attacks}

In this section, we review the state estimation and BDD techniques and the stealth data attacks problem. 

\subsection{State Estimation}

The power system we consider has $n+1$ buses and $n_{t}$ transmission lines. 
The data collected by RTUs includes line power flow and bus power injection measurements. These $m$ 
measurements are denoted by $z = [z_{1},\ldots,z_{m}]^{T}$. The system state $x$ is the vector of phase angles and voltage magnitudes at all buses except the reference bus whose phase angle is set to be zero. For the analysis of cyber security and bad data detection in SE, it is customary to describe the dependencies of measurements and system state through an approximate model called DC power flow model \cite{TeixeiraSou2015}. In the DC power flow 
model, all the voltage magnitudes are assumed to be constant and the reactive power is completely neglected. Thus the vector $z$ refers to active power flow and injection measurements, and the state $x$ refers to bus phase angles only. There are $n$ phase angles to be estimated excluding the reference one, i.e. $x = [x_{1},\ldots,x_{n}]^{T}$. Hence, $z$ and $x$ are related by the equation
\begin{equation}
\label{eq: DC_SE}
z = P\left[\begin{matrix} WB^T\\ -WB^T\\ B_{0} W B^T\\ \end{matrix}\right]x + e := Hx + e,
\end{equation}
where $e \sim \mathcal{N}(0, R)$ is the measurement 
noise vector of independent zero-mean Gaussian variables with the covariance matrix $R = \mbox{diag}(\sigma_{1}^{2}, \ldots, \sigma_{m}^{2})$, $H \in \mathbb{R}^{m \times n}$ represents the system 
model, depending on the topology of the power network, the line parameters and the placement of RTUs. Here the topology is described by a directed incidence matrix $B_{0} \in \mathbb{R}^{(n+1) \times n_t}$ in which the directions of the lines can be arbitrarily specified \cite{TeixeiraSou2015}. Matrix $B \in \mathbb{R}^{n \times n_t}$ is the truncated incidence matrix with the row in $B_{0}$ corresponding to the reference bus removed. The line parameters are described by a diagonal matrix $W \in \mathbb{R}^{n_t \times n_t}$ with diagonal entries being the reciprocals of transmission line reactance. Matrix $P \in \mathbb{R}^{m \times (2n_t+n+1)}$ is a matrix stacked by the rows of identity matrices, indicating which 
power flows or bus 
injections are measured. 
Usually a large degree of redundancy of measurements is employed to make $H$ full rank. 


The state estimate $\hat{x}$ is obtained by the following weighted least squares (WLS) estimate:
\begin{equation}
\label{eq: WLS_DCSE} 
\hat{x} := \mbox{arg} \min\limits_{x} (z - Hx)^{T} R^{-1}(z - Hx),
\end{equation}
which can be solved as $\hat{x} = (H^{T}R^{-1}H)^{-1}H^{T}R^{-1}z := Kz.$
%
%

The estimated state $\hat{x}$ can be used to estimate the active power flows and injections by
\begin{equation}
\label{eq: hat_matrix} 
\hat{z} = H\hat{x} =HKz := Tz,
\end{equation}
where $T$ is the so-called hat matrix \cite{AburExposito2004}. The BDD scheme uses such estimated measurements to identify bad data by comparing $\hat{z}$ with $z$, see below.  

\subsection{Bad Data Detection}

Measurement data may be corrupted by random errors. 
Thus 
there is a built-in BDD scheme in EMS for bad data detection. 
The BDD is achieved by hypothesis tests 
using the statistical properties of the measurement residual:
\begin{equation}
\label{eq: residual}
r = z - \hat{z}=(I - T)z:=Sz=Se, 
\end{equation}
where $ r \in \mathbb{R}^{m}$ is the residual vector, 
$I \in \mathbb{R}^{m \times m}$ is an identity matrix and $S$ is the so-called residual sensitivity matrix \cite{AburExposito2004}.

We now introduce the $J(\hat{x})$-test based BDD. For the measurement error $e \sim \mathcal{N}(0 ,R)$, the new random variable $y=\sum\limits_{i}^{m} R_{ii}^{-1} e_{i}^{2}$ where $R_{ii}$ is the diagonal entry of the covariance matrix $R$ has a $\chi^{2}$ distribution with $m-n$ degrees of freedom. 
Note the quadratic cost function $J(\hat{x}) = \lVert R^{-1/2}r \rVert_{2}^{2}=\lVert R^{-1/2}Se \rVert_{2}^{2}$. 
For the independent $m$ measurements 
we have rank$(S)=m-n$, which implies that 
$J(\hat{x})$ has a so-called \emph{generalized chi-squared distribution} with $m-n$ degrees of freedom \cite{Jones1983}. The BDD 
uses the quadratic function as an approximation of $y$ and checks if it follows 
the distribution $\chi_{m-n}^{2}$. Defining $\alpha \in [0,1]$ as the significance level corresponding to 
the false alarm rate, 
and $\tau(\alpha)$ such that       
\begin{equation}\label{dect_alpha}
\int_{0}^{\tau(\alpha)} {f(x)dx} = 1 - \alpha, 
\end{equation}
where $f(x)$ is the probability distribution function (PDF) of $\chi_{m-n}^{2}$. Hence, the BDD scheme becomes
\begin{equation}
\label{bdd_test}
\left\{
\begin{array}{rcl}
&\text{Good data, if}& \lVert R^{-1/2}r \rVert_{2} \leq \sqrt{\tau(\alpha)}, \\
&\text{Bad data, if}& \lVert R^{-1/2}r \rVert_{2} > \sqrt{\tau(\alpha)}, \\
\end{array}
\right.
\end{equation}
%


\subsection{Stealth FDI Attacks}

The goal of an attacker is to perturb the SE while remaining hidden from the BDD. 
If only data integrity attacks are considered, the attacker could inject false data on a set of measurements, modifying the measurement vector $z$ into $z_{a}:= z+a$. Here the \emph{FDI attack vector} $a \in \mathbb{R}^{m}$ is the corruption added to the original measurement $z$. We have the following definition of a $k_{a}$-tuple FDI attack,
\begin{definition}[$k_{a}$-tuple FDI attack]
\label{def_FDI}
An attack with an FDI attack vector $a \in \mathbb{R}^{m}$ is called a $k_{a}$-tuple FDI attack if 
a number of $k_{a}$ measurements are injected with false data, i.e. $\lVert a \rVert_{0} = k_{a}$. 
\end{definition}

As shown in \cite{LiuNingReiter2009}, an attacker with full knowledge of the system model (i.e., the matrix $H$) and the capability to corrupt specific measurements can keep steath 
if the FDI attack vector follows $a=Hc$ where $c \in \mathbb{R}^{n}$ is non-zero. The corrupted measurements 
$z_{a}$ becomes $z_{a} = H(x+c)+e$. This leads to the state estimate perturbed by a degree of $c$, while the 
residual for BDD checking remains the same. 
It has been verified that such \emph{stealth FDI attacks} based on the DC 
model can be performed on a real SCADA/EMS testbed avoiding the 
bad data detection 
with full nonlinear AC power flow model \cite{TeixeiraDanSandbergEtAl2011}. 

To describe the vulnerability of SE to stealth FDI attacks, the security index is introduced as the minimum number of measurements that need to be corrupted by the attacker in order to keep stealth 
\cite{Sandberg2010}. 
The security index 
is given by  
\begin{align}\label{sec_idx_o}
\begin{split}
\alpha_{j}:= &\min\limits_{c} \quad \lVert a \rVert_{0} \\
&\begin{array}{r@{\quad}r@{}l@{\quad}l}
\mbox{s.t.}\quad a= Hc,\quad a(j) = \mu,\\
\end{array}
\end{split}
\end{align}
where $a(j)$ denotes the injected false data on measurement $j$, and $\mu$ is the non-zero attack magnitude determined by the attacker. The result $\alpha_{j}$ is the security index that quantifies the vulnerability of measurement $j$ to stealth FDI attacks. Here the computed $\alpha_{j}$ belongs to one of the FDI attacks with the minimum $k_{a}$ ($k_{a} = \alpha_{j}$) for measurement $j$. It is known that this optimization problem above is NP-hard (See \cite{HendrickxJohansson2014}). In \cite{TeixeiraSou2015}, the authors proposed an approach using the big M method to express \eqref{sec_idx_o} as a mixed integer linear programming (MILP) problem which can be solved with an appropriate solver,  
\begin{subequations}\label{sec_idx_bigM}   
\begin{align}
\alpha_{j}:= \min\limits_{c,y}\quad& \sum\limits_{i=1}^{m} y(i) \nonumber\\
\mbox{s.t.}\quad& Hc \leq My,\label{eq:10a}\\
 & -Hc \leq My,\label{eq:10b}\\
 & H(j,:)c = \mu,\label{eq:10c}  \\
 & y(i) \in \{0,1\}\quad\mbox{for all } i. \nonumber
\end{align}
\end{subequations} 

In \eqref{sec_idx_bigM}, $M$ is a constant scalar that is greater than the maximum absolute value of entries in $Hc^{*}$, for some optimal solution $c^{*}$ of \eqref{sec_idx_o}. At optimality, for any $i$ that $|H(i,:)c^{*}|=0$, the corresponding $y(i)$ is zero. Thus an optimal solution to \eqref{sec_idx_bigM} is exactly the same optimal solution to \eqref{sec_idx_o} with $y(i) = 1$ indicating that the measurement $i$ is corrupted by an FDI attack. 

\section{Stealth Combined Data Attacks}

FDI attacks 
are resource-intensive since the adversary needs to coordinate integrity attacks on a 
specific number of 
measurements. 
This usually gives the adversary more power than possible in practice \cite{KosutJiaThomasEtAl2011}. 
In reality, an attacker would try to reduce the attack resources and would prefer data availability attacks (e.g., DoS attacks, jamming attacks) since monitoring systems are always more vulnerable to this type of attacks \cite{Markovic2013}. Thus, we focus on the scenario that the adaversary would launch combined data integrity and availability attacks. 

\subsection{Combined Data Integrity and Availability Attacks}

For a large-scale SCADA system, missing data and failing RTUs 
are common \cite{Sandberg2010}. When some of the measurements are missing, 
the typical solution widely employed widely in SE 
is to use the remaining data 
before the system becomes ``unobservable''. Another solution 
is to use 
pseudo measurements (e.g., previous data, forecast information), but 
these measurements would still lose confidence in further time intervals as long as the availability attacks continue. The combined attacks we introduce here are assumed not to make system unobservable and lead to non-convergence of the SE algorithm but try to keep stealth against the BDD. Thus we keep the assumption in this paper that SE uses remaining data if availability attacks take place. We introduce the \emph{availability attack vector} $d \in \{0,1\}^{m}$ for the 
availability attacks and $d(i) = 1$ means that measurement $\mathit{i}$ is unavailable. 
Thus the model for remaining measurements and 
system state can be described by 

\begin{equation}
\label{eq: DC_SE_com}
z_{d} = H_{d}x + e_{d},
\end{equation}
where $e_{d} \in \mathbb{R}^{m}$ and $z_{d} \in \mathbb{R}^{m}$ are the noise vector and measurement vector respectively, and the entries of them are zero if the corresponding measurements are unavailable. Matrix $H_{d} \in \mathbb{R}^{m \times n}$ denotes the model of the remaining measurements and it is obtained from $\mathit{H}$ by replacing some rows with zero row vectors due to availability attacks on these measurements, i.e. $H_{d} :=(I - \mbox{diag}(d)) H$. We can further obtain the hat matrix and residual sensitivity matrix when availability attacks occur, 
\begin{equation}
\label{eq: K0} 
K_{d} := (H_{d}^{T}R^{-1} H_{d})^{-1}H_{d}^{T}R^{-1},
\end{equation}
\begin{equation}
\label{eq: T0} 
T_{d} := H_{d} K_{d}, \quad S_{d} := I  - T_{d}. 
\end{equation}

For the combined attacks, 
the attacker would still launch FDI attacks on the remaining measurements in concert with availability attacks, 
making 
$z_{d}$ 
changed into $z_{a,d}:= z_{d} + a$. Similarly, a $(k_{a}, k_{d})$-tuple combined attack can be defined as
\begin{definition}[$(k_{a}, k_{d})$-tuple combined attack]
\label{def_Com}
A combined attack with an FDI attack vector $a \in \mathbb{R}^{m}$ and an availability attack vector $d \in \{0,1\}^{m}$ described above is called a $(k_{a}, k_{d})$-tuple combined attack if 
$\lVert a \rVert_{0} = k_{a}$, $\lVert d \rVert_{0} = k_{d}$. 
\end{definition}

\subsection{Security Index for Combined Attacks}

Similar to the FDI attacks, if the attack vectors of a $(k_{a}, k_{d})$-tuple attack satisfy $a=H_{d}c$, such combined attacks can still keep stealth as the FDI attack vector $\mathit{a}$ lies on the column space of the matrix $\mathit{H}_{d}$. Using the formulation of security index in \eqref{sec_idx_o} for FDI attacks, we propose an intuitive security index for combined attacks as the minimum number of measurements that need to be compromised by the attacker, 
\begin{subequations}\label{sec_idx_com_o}
\begin{align}
\beta_{j}:= \min\limits_{c,d}\quad& \lVert a \rVert_{0} + \lVert d \rVert_{0} \nonumber\\
\mbox{s.t.}\quad& a = H_{d}c, \label{eq:14a}\\
 & H_{d} =(I - \mbox{diag}(d))H, \label{eq:14b}\\ 
 & a(j) = \mu, \label{eq:14c}\\
 & d(i) \in \{0,1\}\quad\mbox{for all } i. \nonumber
\end{align}
\end{subequations}

Here we also assume $a(j) = \mu$ where $\mu$ is the non-zero attack magnitude. 
The result $\beta_{j}$ is the security index 
that quantifies how vulnerable measurement $j$ is to combined attacks. The computed $\beta_{j}$ belongs to one of the combined attacks that have minimum $k_{a}+k_{d}$ ($k_{a}+k_{d} = \beta_{j}$) for measurement $j$. To solve this NP-hard problem above, we propose a computation solution which uses the big M method to formulate a MILP problem: 
\begin{subequations}\label{sec_idx_com_bigM} 
\begin{align}
\beta^{'}_{j}:= \min\limits_{c,w,d}\quad &\sum\limits_{i=1}^{m} w(i) + \sum\limits_{k=1}^{m} d(k) \nonumber\\
\mbox{s.t.}\quad& Hc \leq M(w+d), \label{eq:15a}\\
 & -Hc \leq M(w+d),  \\
 & H(j,:)c = \mu,  \label{eq:15c}\\
 & w(i) \in \{0,1\}\quad\mbox{for all } i, \label{eq:15d}\\
 & d(k) \in \{0,1\}\quad\mbox{for all } k, \label{eq:15e}
\end{align}
\end{subequations}
where $w, d \in \{0,1\}^{m}$ with $w(i) = 1$ and $d(k) = 1$ meaning FDI attack and data availability attack on measurement $i$ and $k$. 

The following theorem shows that the optimal solution to \eqref{sec_idx_com_o} can be obtained from the optimal solution of \eqref{sec_idx_com_bigM}.
\begin{theorem}\label{theo1_com_sec}
For any 
index $j\in\{1,\dots,m\}$ and non-zero $\mu$, let ($c^{*}$, $w^{*}$, $d^{*}$) be an optimal solution to \eqref{sec_idx_com_bigM}. Then an optimal solution to \eqref{sec_idx_com_o} can be computed as ($c^{*}$, $d^{*}$), and $\beta^{'}_{j} = \beta_{j}$.
\end{theorem}
\begin{proof}
The proof follows by re-writing \eqref{sec_idx_com_o} as \eqref{sec_idx_com_bigM}. First, note that the constraint of~\eqref{sec_idx_com_o}, $a=(I-\mbox{diag}(\mathit{d}))Hc$, can be formulated as a set of inequality constraints with auxiliary binary variables by using the big M method, yielding $ -M w \leq (I-\mbox{diag}(\mathit{d}))Hc \leq M w $, where $w\in \{0,1\}^{m}$ and $\| a \|_0 = \sum{w(i)}$. Since $d$ is a vector of binary variables, the pair of inequality constraints pertaining the $i$-th measurement can be written as $ |(1-d(i)) H(i,:)c| \leq M w(i)$. The latter can be read as
\begin{equation*}
\left\{ \begin{array}{ll}
H(i,:)c = 0, & \mbox{if } w(i)=d(i)=0,\\
|H(i,:)c| \leq M, & \mbox{if } w(i) = 1\mbox{ or } d(i)=1,
\end{array}
\right.
\end{equation*}
which can be rewritten as $|H(i,:)c |\leq M(d(i) + w(i))$. Hence, recalling that $a(i) = (1-d(i)) H(i,:) c$, we conclude that the constraints of \eqref{sec_idx_com_o} can be equivalently re-written as the constraints of \eqref{sec_idx_com_bigM}. The proof concludes by noting that the objective functions of both problems satisfy the equality $\|a\|_0+\|d\|_0 = \sum w(i) + \sum d(i)$.
\end{proof}
%
\begin{corollary}\label{Co1_FDIvsCom}
For any 
index $j\in\{1,\dots,m\}$ and non-zero $\mu$, let ($c^{*}$, $w^{*}$, $d^{*}$) be an optimal solution to \eqref{sec_idx_com_bigM}. Then an optimal solution to \eqref{sec_idx_o} can be computed as $c^{*}$, and $\alpha_{j} = \beta_{j}$. 
\end{corollary}
\begin{proof}
The proof follows straightforwardly from Theorem~\ref{theo1_com_sec}, which establishes that an optimal solution to \eqref{sec_idx_com_o} can be obtained from an optimal solution to \eqref{sec_idx_com_bigM}: comparing \eqref{sec_idx_com_bigM} and \eqref{sec_idx_bigM}, we can easily see that an optimal solution to \eqref{sec_idx_bigM} can be computed as ($c^{*}$, $y^{*}$) with $y^{*}= w^{*} + d^{*}$, and $\alpha_{j}=\beta^{'}_{j}$. Since \eqref{sec_idx_bigM} provides the exact solution to \eqref{sec_idx_o}, an optimal solution to \eqref{sec_idx_o} can be computed as $c^{*}$, and also $\alpha_{j} =\beta^{'}_{j} = \beta_{j}$.
\end{proof}
Corollary~\ref{Co1_FDIvsCom} implies that a set of compromised measurements is an optimal solution to 
\eqref{sec_idx_com_o} if and only if this set is an optimal solution to 
\eqref{sec_idx_o}, and the two security indexes $\beta_{j}$ and $\alpha_{j}$ coincide. 
In fact, in \cite{Sou2013} it was shown that the set of compromised measurements in a $k_a$-tuple FDI attack obtained by solving 
\eqref{sec_idx_o} is a sparsest \emph{critical tuple} containing the target measurement $j$. A sparsest critical tuple is characterized by the measurements that do not belong to a critical tuple of lower order. A critical tuple contains a set of measurements, where removal all of them will cause the system to be unobservable. If any subset of the critical tuple is removed, it would not lead to the loss of observability 
\cite{AburExposito2004}. According to Corollary~\ref{Co1_FDIvsCom} and its proof, we can see that the set of compromised measurements of FDI attacks in this critical tuple is also an optimal solution to the security index problem \eqref{sec_idx_com_o} of combined attacks. The interpretation of the security index problem 
as a critical tuple problem 
provides the means for comparing security indexes of attacks with full and limited adversarial knowledge; 
see Section IV-C for details.

The security indexes derived so far in \eqref{sec_idx_o} and \eqref{sec_idx_com_o} could identify the compromised measurements set of attacks but 
did not consider the attack costs. 
In what follows, we include the costs in the formulation. 
To simplify the discussion, we assume that the 
availability and integrity attacks have the 
costs $\mathit{C}_{A}$ and $\mathit{C}_{I}$, respectively, per measurement. The worst case for power grids is that the adversary succeeds with minimum attack resources. Under these attack costs, we formulate a security index for attack resources of combined attacks as
\begin{equation} 
\label{sec_idx_com_cost}  
\begin{aligned}
\gamma_{j}^{a,d}:=&\min\limits_{c,w,d}\quad \sum\limits_{i=1}^{m} \mathit{C}_{I}w(i) + \sum\limits_{k=1}^{m} \mathit{C}_{A}d(k)\\
\mbox{s.t.} & \quad\eqref{eq:15a}-\eqref{eq:15e}.
\end{aligned}
\end{equation}

By making vector $d$ in \eqref{sec_idx_com_cost} to be zero, we can get the security index $\gamma_{j}^{a}$ for FDI attacks.  
We can also see that the set of compromised measurements from the optimal solution of \eqref{sec_idx_com_cost} is also the optimal solution to \eqref{sec_idx_com_o} and \eqref{sec_idx_o}. 
As previously discussed, it is reasonable to assume that availability attacks can cost less attack resources compared with integrity attacks. If we take the values that satisfy $\mathit{C}_{A} < \mathit{C}_{I}$, the optimal solution of $w^{*}$ and $d^{*}$ in \eqref{sec_idx_com_cost}, w.r.t. measurement $j$, would lead to $\sum w^{*}(i) = 1$ and $\sum d^{*}(k) = \beta_j-1$. This means that the optimal combined attack in the case of $\mathit{C}_{A} < \mathit{C}_{I}$ is to corrupt one measurement with an integrity attack and make other measurements in this critical tuple unavailable. This statement is made formal in the following proposition which will be validated in Section VI-A.
\begin{proposition}\label{pro1_optcost_attack}
When $\mathit{C}_{A} < \mathit{C}_{I}$, the optimal strategy of combined attack is to inject false data on the targeted measurement $j$ and make other measurements in the critical tuple unavailable to the SE, yielding a (1,$\beta_j - 1$)-tuple combined attack with optimal attack cost $\gamma_{j}^{a,d} = C_I + (\beta_j - 1)C_A$. 
\end{proposition}
%


\section{Attacks with Limited Adversarial Knowledge}

In this section we consider the scenario in which the adversary has limited knowledge of the system model and discuss how this affects the detectability of combined attacks. 

\subsection{Relaxing Assumption on Adversarial Knowledge}

For the combined attacks 
above, the adversary is assumed to have full knowledge of 
$H$ in \eqref{eq: DC_SE} that includes the topology of the power network, the placement of RTUs and the transmission line reactance. %
This 
system 
data is usually kept in the database of control center, which is 
difficult to be accessed by the attacker. We extend the previous analysis by replacing the full knowledge assumption. Hence, in what follows the attacker only has limited knowledge of the system model. 
An attacker could acquire limited knowledge as a result of 
analyzing 
an out-dated or estimated model using power network topoloy data but limitd information of transmission line parameters \cite{Teixeira2010}\cite{Rahman2012}.

Looking at the problem from the attacker's perspective, we denote the perturbed system model as $\tilde{H}$, such that 
\begin{equation}
\label{err_H}
\tilde{H} = H + \Delta H, 
\end{equation}
where $\Delta H \in \mathbb{R}^{m \times n}$ denotes the part of model uncertainty. We still consider that the attacker uses the same linear policies to compute attack vectors, i.e. $a=\tilde{H}_{d}c$ for combined attacks and $a=\tilde{H}c$ for FDI attacks 
and $\tilde{H}_{d} := (I-\mbox{diag}(d))\tilde{H}$. 
  
\subsection{Detectability of Data Attacks}

\subsubsection{Combined Attacks}

When the measurements are corrupted by a $(k_{a}, k_{d})$-tuple attack, the measurement residual $r(a,d)$ can be written as
\begin{equation}\label{residual_attack}
r(a,d) = S_{d}z_{a,d}=S_{d}e_{d} + S_{d}a.
\end{equation}

As discussed in Section III-B, when the attack vectors of the 
combined attack 
satisfy $a=H_{d}c$, the residual $r(a,d) = S_{d}e_{d} + S_{d}H_{d}c = S_{d}e_{d}$ due to $S_{d}H_{d}=0$, then the residual is not affected by $a$ and no additional alarms are triggered; the BDD treats the measurements attacked by availability attacks as a case of missing data. However, for the attack with limited knowledge, 
the attack vector $a$ becomes $a=\tilde{H}_{d}c$ and $S_{d}a$ may be non-zero. In this case,  
the 
residual is incremented and the attack can be detected with some possibility. 
 
Note that the quadratic cost function with the 
combined attack becomes $J_{a,d}(\hat{x}) = \lVert R^{-1/2}S_{d}e_{d} + R^{-1/2}S_{d}a \rVert_{2}^{2}$. 
Here the mean of $( R^{-1/2}S_{d}e_{d} + R^{-1/2}S_{d}a)$ is non-zero $R^{-1/2}S_{d}a$ incremented by the attack.  
Recalling the $J(\hat{x})$-test based BDD, 
$J_{a,d}(\hat{x})$ has a \emph{generalized non-central chi-squared distribution} with $m - n - k_{d}$ degrees of freedom under the combined attack. 
We use $J_{a,d}(\hat{x})$ as an approximation of having the \emph{non-central chi-squared distribution} $\chi_{m-n-k_{d}}^{2}(\lVert R^{-1/2}S_{d}a \rVert_{2}^{2})$ to calculate the detection probability, where $\lambda_{a,d} = \lVert R^{-1/2}S_{d}a \rVert_{2}^{2}$ is the non-centrality parameter. Further we will validate such approximation using empirical results from Monte Carlo simulation in Section VI-B. 
We can further obtain
\begin{equation}\label{nchi_pdf_com}
\int_{0}^{\tau_{d}(\alpha)} {f_{\lambda_{a,d}}(x)dx} = 1 - \delta_{a,d}, 
\end{equation}
where $f_{\lambda_{a,d}}(x)$ is the PDF of $\chi_{m-n-k_d}^{2}(\lVert R^{-1/2}S_{d}a \rVert_{2}^{2})$, $\tau_{d}(\alpha)$ is the threshold set in the BDD using \eqref{dect_alpha} but with the PDF of $\chi_{m-n-k_{d}}^{2}$, 
and $\delta_{a,d}$ is the detection probability. 

\subsubsection{FDI Attacks}

For a $k_{a}$-tuple FDI attack with limited knowledge, the quadratic function $J_{a}(\hat{x})$ can also be approximated to have a non-central chi-squared distribution but with $m-n$ degrees of freedom, namely the distribution $\chi_{m-n}^{2}(\lVert R^{-1/2}Sa \rVert_{2}^{2})$. Similar to \eqref{nchi_pdf_com}, the detection probability can be computed by solving
\begin{equation}\label{nchi_pdf}
\int_{0}^{\tau(\alpha)} {f_{\lambda_{a}}(x)dx} = 1 - \delta_{a}, 
\end{equation}
where $\lambda_{a} = \lVert R^{-1/2}Sa \rVert_{2}^{2}$ denotes the non-centrality parameter,  
$\tau(\alpha)$ is the threshold set in the BDD using \eqref{dect_alpha}, 
and $\delta_{a}$ is the detection probability of the FDI attack. 

\subsection{Special Case: Attacks with Structured Model Uncertainty}

An interesting analysis can be made to understand what the model uncertainty $\Delta H$ is to the adversary. As stated in \cite{Teixeira2010}, 
the scenarios where the uncertainty is more structured are of greater interest. Here we assume that the attacker knows the exact topology of the power network and the placement of RTUs, but has to estimate the line parameters. This assumption is feasible since the attacker can analyze the topology according to the breaker status data and compute the model based on available power flow measurements, 
while usually the attacker has limited access to the knowledge of the exact length of the transmission line and type of the conductor being used \cite{Rahman2012}. Thus the 
model with such structured uncertainty becomes
\begin{equation}
\label{structured_uncertainty}
\tilde{H}=P \left[\begin{matrix} \tilde{W}B^T\\ -\tilde{W}B^T\\ B_{0} \tilde{W} B^T\\ \end{matrix}\right].
\end{equation}
where $\tilde{W}$ is derived from $W$ but with errors. 
Now we consider the security index of attacks w.r.t. $\tilde{H}$ in \eqref{structured_uncertainty}. As we have discussed in Section III-B, the security index problem can be interpreted as a critical tuple problem. 
In the remaining part of this paper we adopt the following assumption,
\begin{assumption}\label{assum1}
The system with perturbed model $\tilde{H}$ in \eqref{structured_uncertainty} has the same sets of critical tuples as the system with original model $H$ in \eqref{eq: DC_SE}.
\end{assumption}

Assumption~\ref{assum1} is expected to hold in the case that the system with $H$ in \eqref{eq: DC_SE} is topologically observable \cite{Krumpholz1980}. Defining the security indexes for compromised measurements set under structured uncertainty model as $\tilde{\alpha}_j$ and $\tilde{\beta}_j$, 
the following theorem 
shows that the security index remains the same although the model 
is perturbed with structured uncertainty.          
\begin{theorem}\label{the2_sec_conicide}
For any measurement index $j\in\{1,\dots,m\}$ and non-zero $\mu$, under Assumption~\ref{assum1}, let ($\tilde{c}^{*}$, $\tilde{w}^{*}$, $\tilde{d}^{*}$) be an optimal solution to \eqref{sec_idx_com_bigM} w.r.t. $\tilde{H}$ in \eqref{structured_uncertainty}. 
Then there exists some $c^{*}$ such that ($c^{*}$, $w^{*}$, $d^{*}$) with $w^{*}=\tilde{w}^{*}$ and $d^{*}=\tilde{d}^{*}$ is an optimal solution to \eqref{sec_idx_com_bigM} w.r.t. $H$ in \eqref{eq: DC_SE}, ($c^{*}$, $y^{*}$) with $y^{*}=\tilde{w}^{*} + \tilde{d}^{*}$ is an optimal solution to \eqref{sec_idx_bigM} w.r.t. $H$ in \eqref{eq: DC_SE}, 
and $\tilde{\beta}_j= \beta_j = \alpha_j = \tilde{\alpha}_j$.  
\end{theorem}    
\begin{proof}
The optimal solution with $\tilde{w}^{*}$ and $\tilde{d}^{*}$ identifies a sparsest critical tuple containing measurement $j$ for the perturbed model $\tilde{H}$ in \eqref{structured_uncertainty}, which is also a sparsest critical tuple for the model $H$ in \eqref{eq: DC_SE} according to Assumption~\ref{assum1}. Then the set of measurements in this critical tuple is an optimal solution to the security index problem of \eqref{sec_idx_com_bigM} w.r.t. $H$ in \eqref{eq: DC_SE}. According to Theorem~\ref{theo1_com_sec} and Corollary~\ref{Co1_FDIvsCom}, the set of measurements in this critical tuple is also an optimal solution to the security index problem of \eqref{sec_idx_bigM} w.r.t. $H$ in \eqref{eq: DC_SE}.   
\end{proof}

With respect to the security index for attack resources, let $\tilde{\gamma}_j^{a,d}$ and $\tilde{\gamma}_j^{a}$ be the security indexes of combined attacks and FDI attacks from \eqref{sec_idx_com_cost} but w.r.t. perturbed model $\tilde{H}$ in \eqref{structured_uncertainty}. 
We can see that the set of compromised measurements from optimal solution to \eqref{sec_idx_com_cost} w.r.t. $\tilde{H}$ in \eqref{structured_uncertainty} is also the optimal solution to \eqref{sec_idx_com_bigM} and \eqref{sec_idx_bigM} according to Theorem~\ref{the2_sec_conicide}. When it is the case that $C_{A} < C_{I}$, the optimal solution of $\tilde{w}^{*}$ and $\tilde{d}^{*}$ from \eqref{sec_idx_com_cost} w.r.t. $\tilde{H}$, would lead to $\sum \tilde{w}^{*}(i) = 1$ and $\sum \tilde{d}^{*}(k) = \tilde{\beta}_j-1$. Such (1,$\tilde{\beta}_j-1$)-tuple combined attack can be launched with least attack resources when $C_{A} < C_{I}$ and in the following we show that it also can achieve minimized detectability. 

As discussed in Section IV-B, the detection probability would increase 
when 
attacker has limited knowledge of the system model. 
However, for the combined attacks, 
the following proposition states that the combined attacks with structured model uncertainty can still keep stealth against the BDD if the following conditions are satisfied: 1) structured model uncertainty is defined as in \eqref{structured_uncertainty}; 2) Assumption~\ref{assum1} holds.  
\begin{proposition}\label{pro2_steal_com}
For any 
index $j\in\{1,\dots,m\}$ and non-zero $\mu$, under Assumption~\ref{assum1}, let ($\tilde{c}^{*}$, $\tilde{w}^{*}$, $\tilde{d}^{*}$) with $\sum \tilde{w}^{*}(i)=1$ be an optimal solution to \eqref{sec_idx_com_bigM} w.r.t. $\tilde{H}$ in \eqref{structured_uncertainty}. Then this (1,$\tilde{\beta}_j-1$)-tuple combined attack 
from ($\tilde{c}^{*}$, $\tilde{w}^{*}$, $\tilde{d}^{*}$) is a stealth attack. 
\end{proposition}
\begin{proof}
The FDI attack vector of this combined attack is $a = \tilde{H}_{\tilde{d}^{*} }\tilde{c}^{*}$. According to Theorem~\ref{the2_sec_conicide}, there exists $c^{*}$ such that ($c^{*}$, $w^{*}$, $d^{*}$) with $w^{*}=\tilde{w}^{*}$ and $d^{*}=\tilde{d}^{*}$ is an optimal solution to \eqref{sec_idx_com_bigM} w.r.t. $H$ in \eqref{eq: DC_SE}. Using the attack strategy above, $k_{a}=\sum \tilde{w}^{*}(i)= 1$ and the only non-zero entry of the attack vector $a$ is $\mu$ while other measurements in this critical tuple are attacked by availability attacks. Thus this combined attack is with the vector $a = (I - \mbox{diag}(\tilde{d}^{*})) \tilde{H}\tilde{c}^{*} = (I - \mbox{diag}(d^{*})) Hc^{*} = H_{d^{*}}c^{*}$, which can keep stealth w.r.t. $H$ in \eqref{eq: DC_SE}.     
\end{proof}
%

\section{Risk Assessment for Data Attacks}

The previous sections focus on vulnerability assessment of SE to combined attacks with limited knowledge. Following the procedure of risk analysis in \cite{Sridhar2012}, in this section we define and analyze the \emph{risk} brought by attacks with limited knowledge. 


Usually the total \emph{risk} of data attacks is defined as the likelihood of attack multiplied by the potential attack impact \cite{Ronald2012}. For a $(k_a, k_d)$-tuple combined attack, the risk metric $\mathbf{R}(a,d)$ can be expressed as 
\begin{equation}\label{risk_def}
\mathbf{R}(a,d) = \mathbf{L}(a,d)*\mathbf{I}(a,d)
\end{equation}
where $\mathbf{L}(a,d)$ denotes the likelihood of the combined attack with attack vectors $a$ and $d$, 
and $\mathbf{I}(a,d)$ denotes the attack impact. 
For the attacks with larger risk metrics, they bring more risk to reliable system operation. In the following we discuss how $\mathbf{L}(a,d)$ and $\mathbf{I}(a,d)$ are formulated.  

\subsection{Likelihood of Data Attacks}

The attack likelihood relates to the vulnerability of the system. In this work, the likelihood of the attack is taken as the probability that the attack is launched and the probability that the attack can keep stealth against the detection schemes,
\begin{equation}\label{likelihood}
\mathbf{L}(a,d) = P(a,d) P(s|a,d),
\end{equation}
where %
$P(s|a,d)$ denotes the conditional probability of the combined attack passing the BDD if it has been performed. 
For the attack with limited knowledge, the detection probability $\delta_{a,d}$ can be obtained from \eqref{nchi_pdf_com}, thus we have $P(s|a,d) = 1 - \delta_{a,d}$. 
In \eqref{likelihood}, $P(a,d)$ represents the probability that a particular adversary would perform a combined attack and successfully corrupt the data. Obtaining meaningful and realistic data for calculating $P(a,d)$ remains an unsolved and open issue for most of the established approaches \cite{Ashok2017}.  
The proposed security index $\tilde{\gamma}_j^{a,d}$ w.r.t. perturbed model $\tilde{H}$ captures the efforts required by a combined attack 
and essentially can be related to the probability $P(a,d)$. We assume that if the attacks 
have the same security index of $\tilde{\gamma}_j^{a,d}$, they have the same probability of $P(a,d)$.  
In this paper, to compare the risk of attacks with the same security index, 
we ``normalize'' $P(a,d)$ to be 1, 
meaning that the attacks have been performed successfully. 
The following risk metric applies to the attacks with the same security index of $\tilde{\gamma}^{a,d}_j$,
\begin{equation}\label{risk}
\mathbf{R}(a, d) = P(a,d) P(s|a,d) \mathbf{I}(a, d) = (1-\delta_{a,d}) \mathbf{I}(a, d),
\end{equation}
%

For the $k_a$-tuple FDI attacks with the same security index of $\tilde{\gamma}_{j}^{a}$, 
the formulation of risk metric is similar, i.e. $\mathbf{R}(a) = (1 - \delta_{a}) \mathbf{I}(a)$ where $\delta_{a}$ is the detection probability from \eqref{nchi_pdf}, $\mathbf{I}(a)$ denotes the attack impact and $\mathbf{R}(a)$ is the risk metric. Thus in the case of $\tilde{\gamma}_{j}^{a,d} = \tilde{\gamma}_{j}^{a}$, the risk of combined attacks and FDI attacks is comparable.   

\subsection{Attack Impact: Errors of Load Estimate}

The estimated information from SE is used by further applications in EMS to compute optimal control actions. These are typically computed by minimizing network operation costs which are obtained by solving OPF algorithms. As the work in \cite{Liang2016} 
shows, the OPF application uses the load estimate as the inputs. 
If data attacks take place and pass the BDD, the load estimates get perturbed which influences the control actions. 
Therefore, we consider the impact metric as a function of the bias introduced by the attack on the load estimate.

Assuming that there are $m_{inj}$ injection measurements including loads, we consider the impact on the errors of estimating net power injections, which can be described as
\begin{equation}\label{err_inj}
\epsilon = \hat{z}_{inj,a,d} - z_{inj},
\end{equation}
where $z_{inj} \in \mathbb{R}^{m_{inj}} $ is the original injection measurements including loads 
and $\hat{z}_{inj, a, d} \in \mathbb{R}^{m_{inj}}$ is the vector of estimated measurements under a $(k_{a}, k_{d})$-tuple combined attack. Thus
\begin{equation}
\epsilon = H_{inj} \hat{x}_{a,d} - (H_{inj}x+e_{inj}),
\end{equation}
where $\hat{x}_{a, d}=K_{d}(z_{d} + a)=x+K_{d}e_{d}+K_{d}a$, $H_{inj} \in \mathbb{R}^{m_{inj} \times n}$ denotes the submatrix of $H$ by keeping the rows corresponding to injection measurements including loads, 
and $e_{inj} \in \mathbb{R}^{m_{inj}}$ is the noise vector of these measurements. We can further obtain $\epsilon = H_{inj} K_{d}a + H_{inj} K_{d}e_{d} -e_{inj}$ where the term introduced by the attacks is $H_{inj} K_{d}a$. Here $K_{d}$ is the function of the matrix $H_{d}$ as defined in \eqref{eq: K0}. The expected value of $\epsilon$ is
\begin{equation}\label{exp_err}
\mathbb{E}(\epsilon) = H_{inj} K_{d}a.
\end{equation}
We have the following definition of the attack impact metric for combined attacks.
\begin{definition}\label{def_attackimpact}
The impact metric $\mathbf{I}(a,d)$ for quantifying attack impact of a combined attack with FDI attack vector $a$ and availability vector $d$ on load estimate is defined as the 2-norm of $H_{inj}K_{d}a$, i.e. $\mathbf{I}(a,d) := \lVert H_{inj} K_{d}a \rVert_{2}$.
\end{definition}

Similar to the combined attacks, we define the attack impact metric $\mathbf{I}(a) = \lVert H_{inj} Ka \rVert_{2}$ for a $k_a$-tuple FDI attack with attack vector $a$. We continue to adopt the linear attack policies to compute attack vectors for attacks with 
limited knowledge, i.e., 
$a=\tilde{H}_{d}c$ for combined attacks and 
$a=\tilde{H}c$ for FDI attacks. 

\section{Case Study}

In this section we apply the analysis to the IEEE 14-bus system 
(\myfigref{Figure_14bussystem2}).  
We conduct simulations on DC 
model for the purposes of: 1) illustrating vulnerability of SE to combined attacks;
2) providing insights into how combined attack can differ from FDI attack; 3) evaluating the risk of data attacks and giving the risk prioritization. In the performed experiments, 
measurements are placed on all the buses and transmission lines to provide large redundancy. The per-unit system is used and the power base is $100 MW$. The 
measurements are generated under the DC model with Gaussian noise ($\sigma_j = 0.02$ for any measurement $j$). For the limited knowledge model, we assume that 
the attacker knows the exact topology 
but has estimated line parameters with errors up to $\pm$20$\%$. %

\myincludegraphics{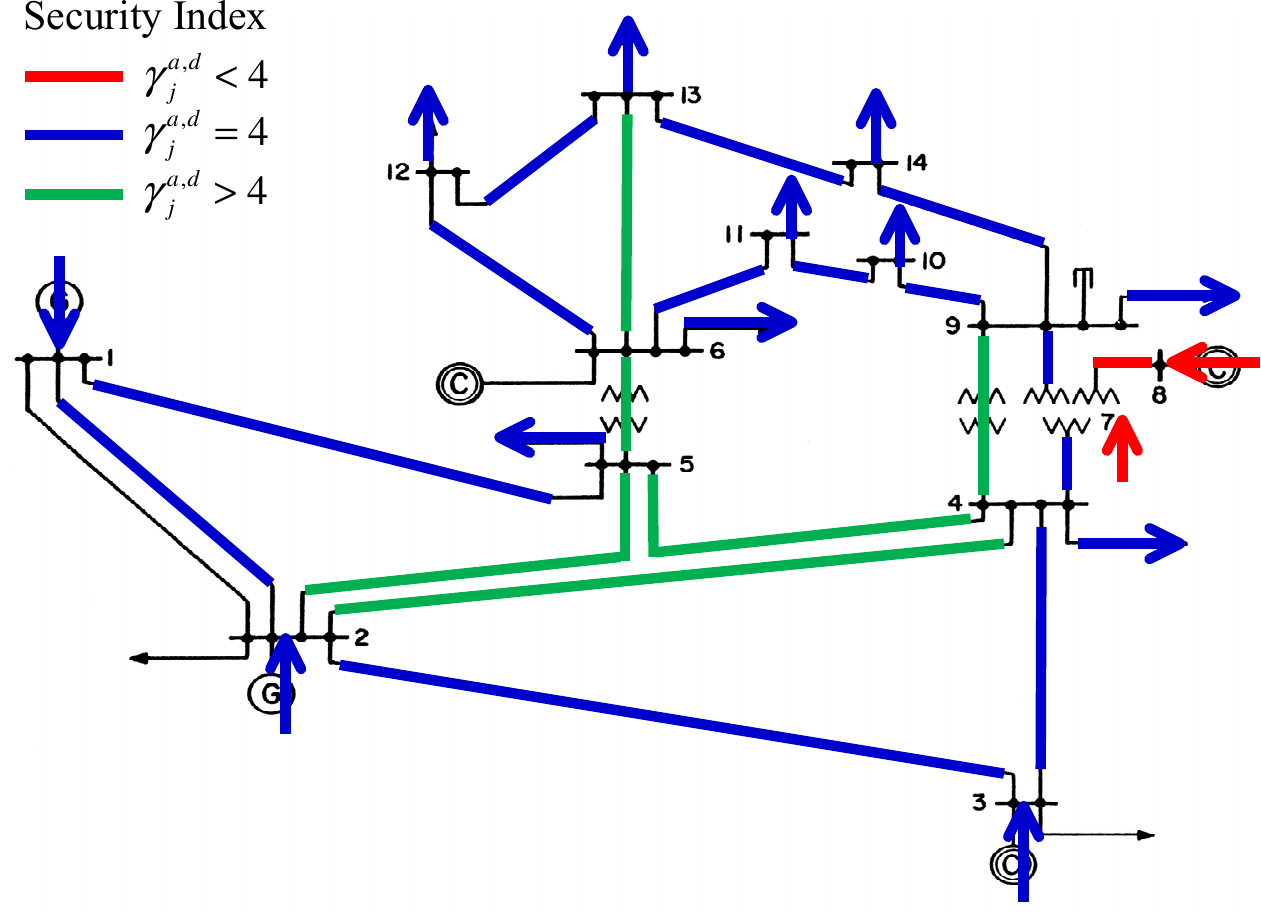}{The IEEE 14-bus 
system. The measurements are labeled different colors according to their security index $\gamma_j^{a,d}$ from \myfigref{Figure1_Securityindex2}. 
The most vulnerable measurements with small
index ($<4$) are color coded red. The measurements that
have large index ($>4$) are color coded green. The others 
are color coded blue and their vulnerabilities
lie somewhere in between. A similar figure of measurements under FDI attacks can be found in \cite{TeixeiraSou2015}.}

\subsection{Security Index for Vulnerability Analysis}  

In order to expose vulnerability of SE to data attacks, we calculated the security index using the computation solutions of \eqref{sec_idx_com_bigM} (according to Theorem~\ref{theo1_com_sec}) and \eqref{sec_idx_bigM} for both combined attacks and FDI attacks. Thus the minimum number of compromised measurements and attack resources needed by the attacker to corrupt SE and pass the BDD are determined. \myfigref{Figure1_Securityindex2} shows the security indexes $\gamma_j^{a,d}$ and $\gamma_j^{a}$ of combined attacks and FDI attacks, where the x-axis indicates the measurement targeted by the attacker to inject false data of $\mu = 0.1 p.u.$. The results illustrate the attack resources needed by the attacker to keep stealth. 
The security index of combined attacks is also showed in \myfigref{Figure_14bussystem2} where the measurements are color coded to indicate which ones are more vulnerable. 
Combining \myfigref{Figure1_Securityindex2} and \myfigref{Figure_14bussystem2}, the security index can illustrate the security week point in a power system.   

\myincludegraphics{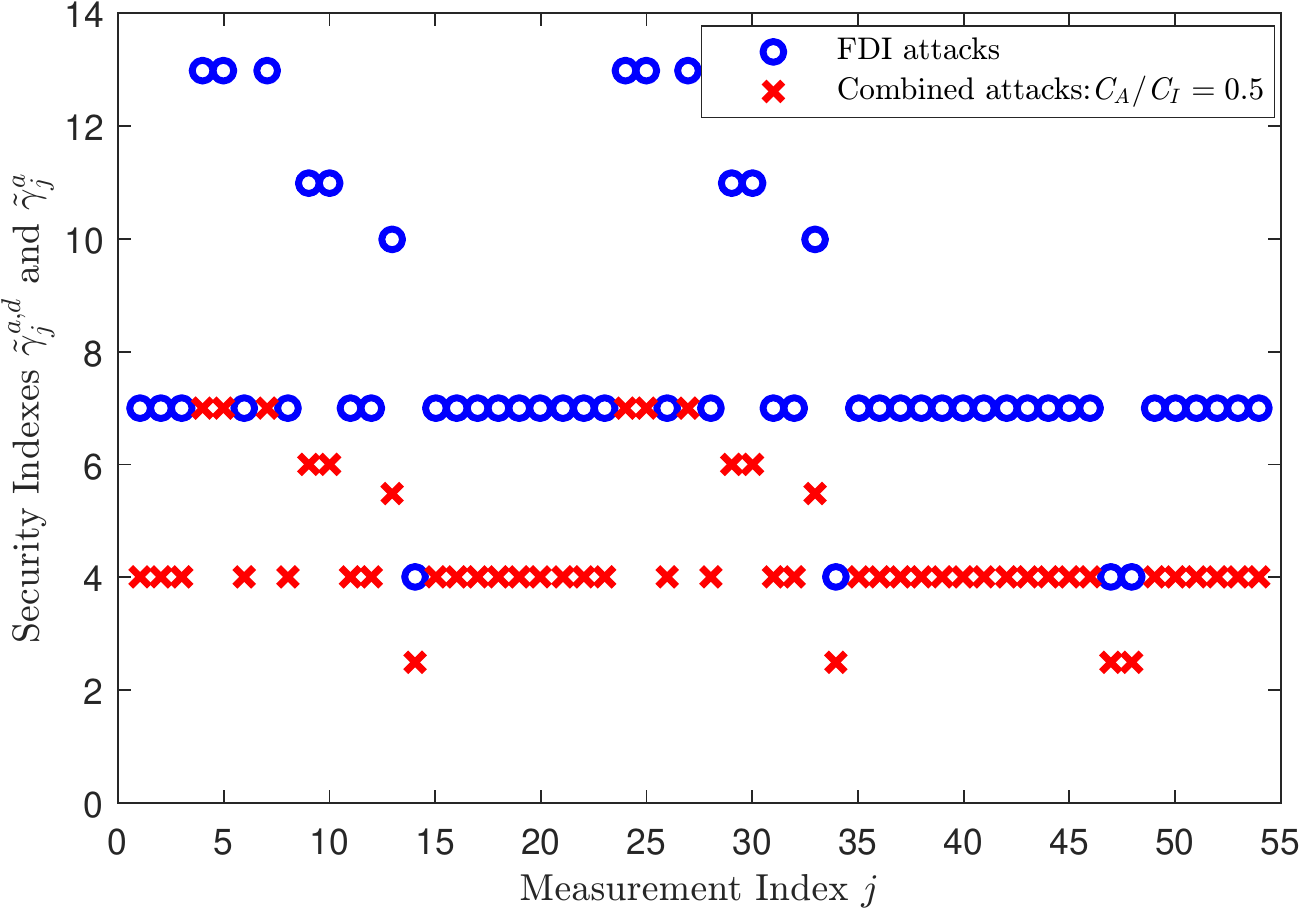}{The security index $\gamma_{j}^{a,d}$ under combined attacks and 
$\gamma_{j}^{a}$ under FDI attacks are plotted versus measurement index $j$. 
Here the cost of FDI attack on per measurement is assumed to be 1 
and $C_{A} = 0.5$ as $C_{A}/C_{I}=0.5$.}

The values of security index under combined attacks are smaller than the ones under FDI attacks when $\mathit{C}_{A} < \mathit{C}_{I}$ 
from \myfigref{Figure1_Securityindex2}. For instance, in order to corrupt measurement $j = 9$, the FDI attack needs a value of 11 for attack resources (i.e. a 11-tuple FDI attack) while the combined attack only needs a value of 6 (i.e. a (1,10)-tuple combined attack). This implies that SE is more vulnerable to combined attacks with less attack resources. 
The results also show that $k_a = 1$ for the combined attacks %
and the optimal attack cost is $C_I + (\beta_j - 1)C_A$ for the case $\mathit{C}_{A} < \mathit{C}_{I}$, 
which is consistent with Proposition~\ref{pro1_optcost_attack}.  

\subsection{Detectability of Attacks with Limited Knowledge}

Using the attack policy $a = \tilde{H}_{d}c$ for combined attacks and $a = \tilde{H}c$ for FDI attacks with the given model uncertainty, the detection probability of attacks 
can be obtained according to \eqref{nchi_pdf_com} and \eqref{nchi_pdf}. From Theorem~\ref{the2_sec_conicide} we see that the compromised measurements set from the optimal solutions of \eqref{sec_idx_com_cost} 
w.r.t. $\tilde{H}$ in \eqref{structured_uncertainty} is in the same critical tuple with the one w.r.t. $H$ in \eqref{eq: DC_SE}. Thus a set of 11 measurements (a critical tuple) containing measurement $j=9$ needs to be compromised by the attacker from the security index in \myfigref{Figure1_Securityindex2}. For the sake of comparison, the combined attacks and FDI attacks are performed in the same set of these 11 measurements. \myfigref{Figure2_Detectionprob2} shows the detection probability of combined attacks and FDI attacks targeting these 11 measurements. 
In addition to the theoretical results,  
the empirical detection probability results are also presented for the 11-tuple FDI attack and (2,9)-tuple combined attack respectively. 

To obtain the empirical detection probability, 
we use Monte Carlo simulations. Taking the (2,9)-tuple combined attack as an example, 
200 different points of attack magnitude $\mu$ were taken in random from 0 to 0.5 p.u. and the corresponding attack vectors were built. For each attack vector with the taken magnitude $\mu$, total 1000 Monte Carlo runs were executed to obtain the 
detection probability of such attack. In each Monte Carlo simulation, the 
measurements were created by the DC model with Gaussian noise and the attack vector was added to the measurements. For the attacked measurements, the SE and BDD with the false alarm rate 0.05 were executed.

\myincludegraphics{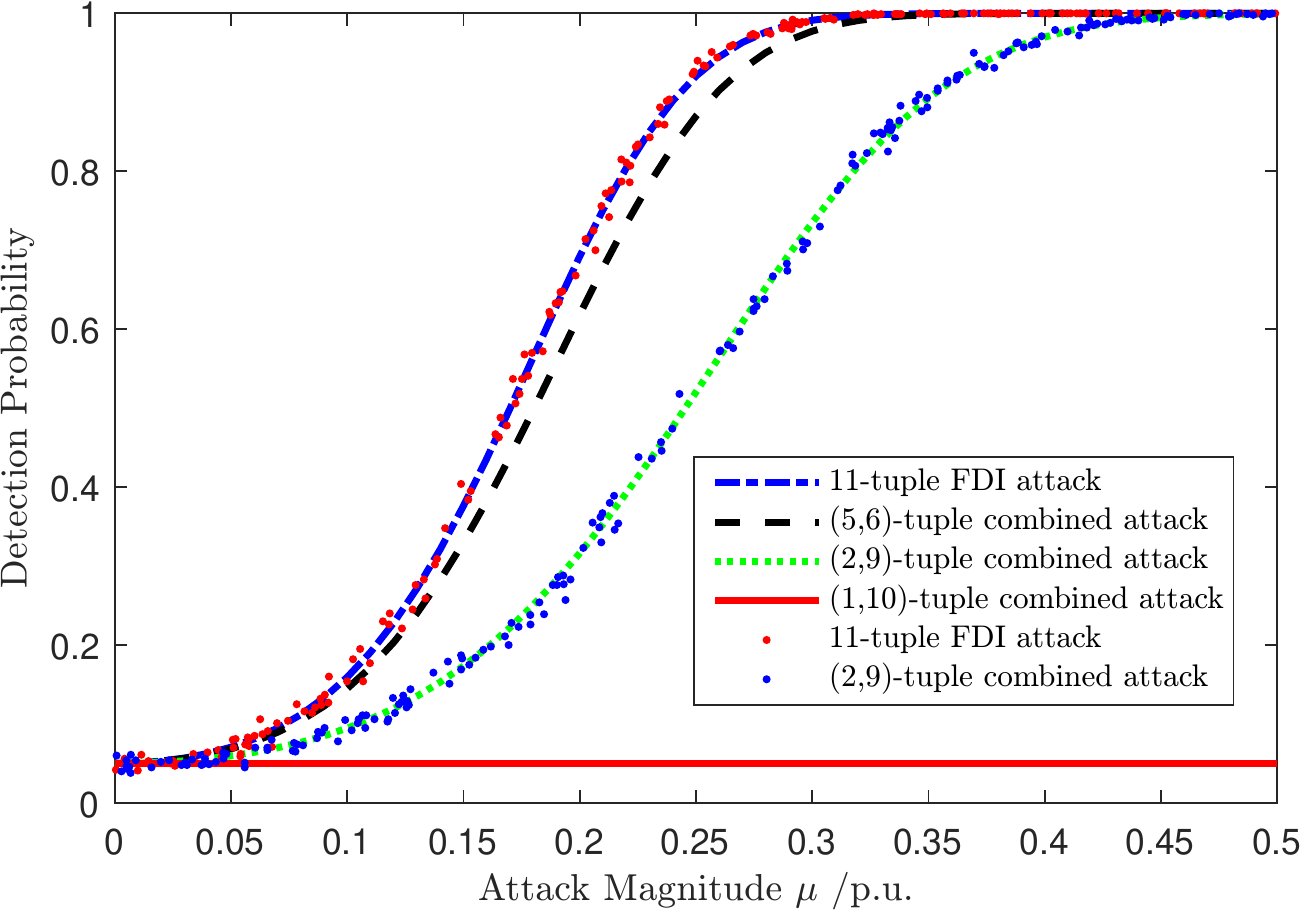}{The detection probability is plotted versus the attack magnitude. 
The attacks are under structured uncertainty model and performed in the 
set of 11 measurements and the 
false alarm rate $\alpha$ is 0.05.} 

From \myfigref{Figure2_Detectionprob2} we can see that the empirical results of detection probability follow the theoretical one. This proves that using the approximation of the distribution of $J_{a,d}(\hat{x})$ and $J_{a}(\hat{x})$ 
can provide the detection probability, and it is reliable to use theoretical detection probability for risk analysis in the following. The results in \myfigref{Figure2_Detectionprob2} illustrate that combined attacks can have lower detection probability comparing with FDI attacks, meaning that SE is more vulnerable to combined attacks as they have higher probability not to be discovered by the BDD. An interesting result is that with smaller $k_a$ the combined attack also has lower probability to be detected. In the case that $k_a = 1$ and $k_d = 10$, the (1,10)-tuple combined attack can keep stealth, 
which is consistent with Proposition~\ref{pro2_steal_com}.     

\subsection{Risk Metrics for Attacks} 

We continue with the risk analysis of combined attacks. %
Simulations were conducted on the same scenarios as Section VI-B where the attacker manipulates the set of 11 measurements (a critical tuple). 
We analyze the attack impact and present the risk of 
the combined attacks and FDI attacks. For the risk analysis, we take the attack cost values that satisfy $C_A = C_I$, thus the security indexes $\tilde{\gamma}_j^{a,d}$ and  $\tilde{\gamma}_j^{a}$ w.r.t. $\tilde{H}$ in \eqref{structured_uncertainty} of these attacks are equal to each other 
and the probability $P(a,d)$ can be ``normalized'' as discussed in Section IV-B. The results for attack impact metrics versus detection probability are given in \myfigref{Figure3_ImpvsDetectionprob2}. The values of risk metrics for combined attacks and FDI attacks are shown in \myfigref{Figure4_Risk2}. 

Under the perturbed model with uncertainty, the attacker has the possibility to be detected by the BDD while introducing errors on load estimate. From \myfigref{Figure3_ImpvsDetectionprob2}, we see that combined attacks can have similar attack impact metrics with FDI attacks but lower detection probability with the same attack magnitude $\mu$ (0.15 p.u. or 0.25 p.u. as shown in \myfigref{Figure3_ImpvsDetectionprob2}). 
Especially the (1,10)-tuple combined attack 
has larger impact metrics than attacks with limited knowledge for the both cases that attack magnitude $\mu = 0.15 p.u.$ or $\mu = 0.25 p.u.$. 

For the \emph{risk} metrics in \myfigref{Figure4_Risk2}, when the attack magnitude $\mu$ increases from zero, the risk metric increases due to the low detection probability. After $\mu$ reaches certain values, the risk metric decreases since the attacks can be discovered 
with high probability. It's also shown that combined attacks can have larger risk metrics especially the cases of (1,10)-tuple and (2,9)-tuple combined attacks. It should be noted that though we assume $C_A = C_I$ to obtain the risk metrics, the risk prioritization of these attacks in \myfigref{Figure4_Risk2} would not change if $C_A < C_I$ is assumed. This is because the combined attacks can be launched with less attack resources 
when $C_A < C_I$, 
resulting in larger risk values comparing with FDI attacks. %

\myincludegraphics{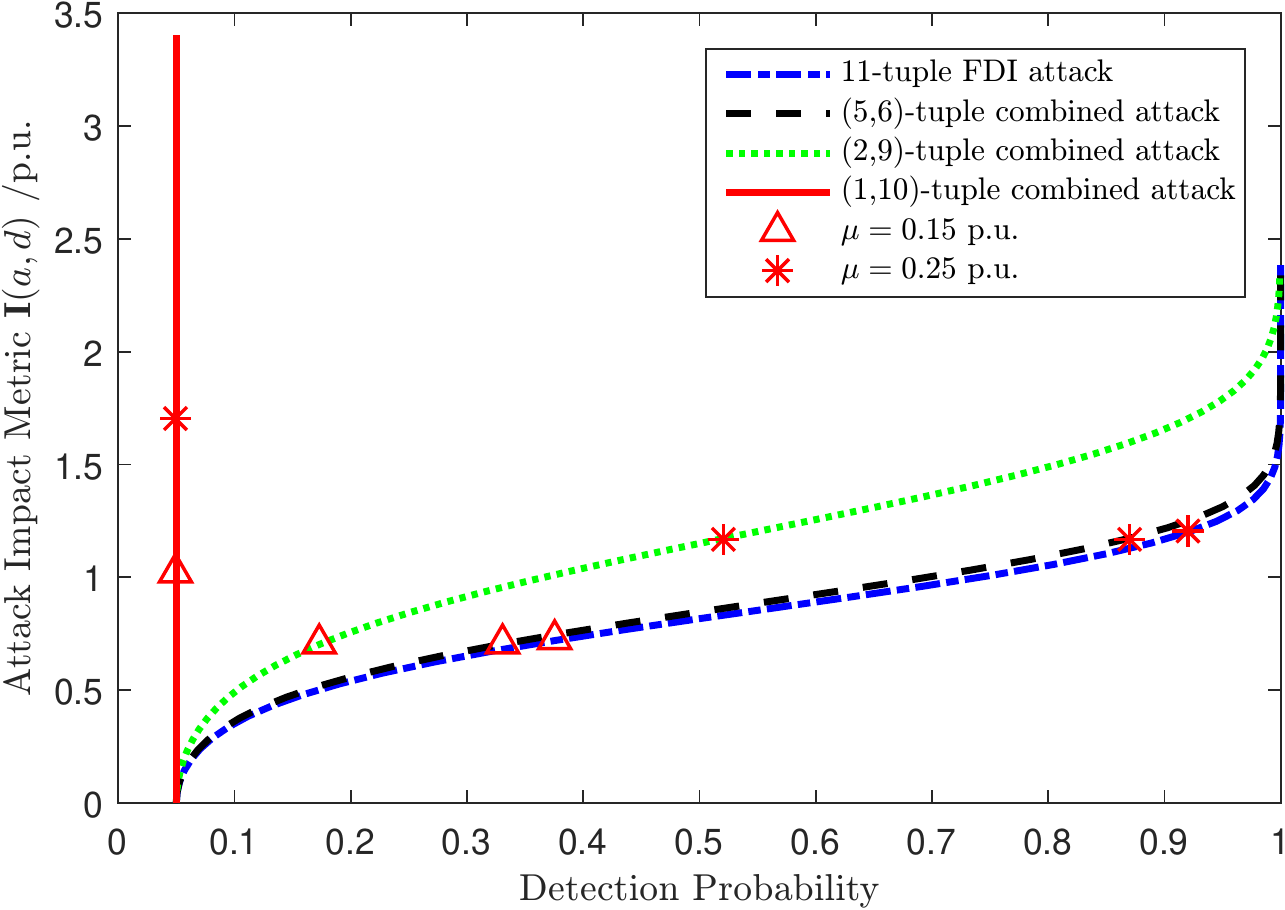}{The attack impact metric is plotted versus the detection probability. The attacks are under structured uncertainty model and performed in the 
set of 11 measurements. 
Here we assume $C_A = C_I$ and 
false alarm rate $\alpha$ is 0.05.} 

\myincludegraphics{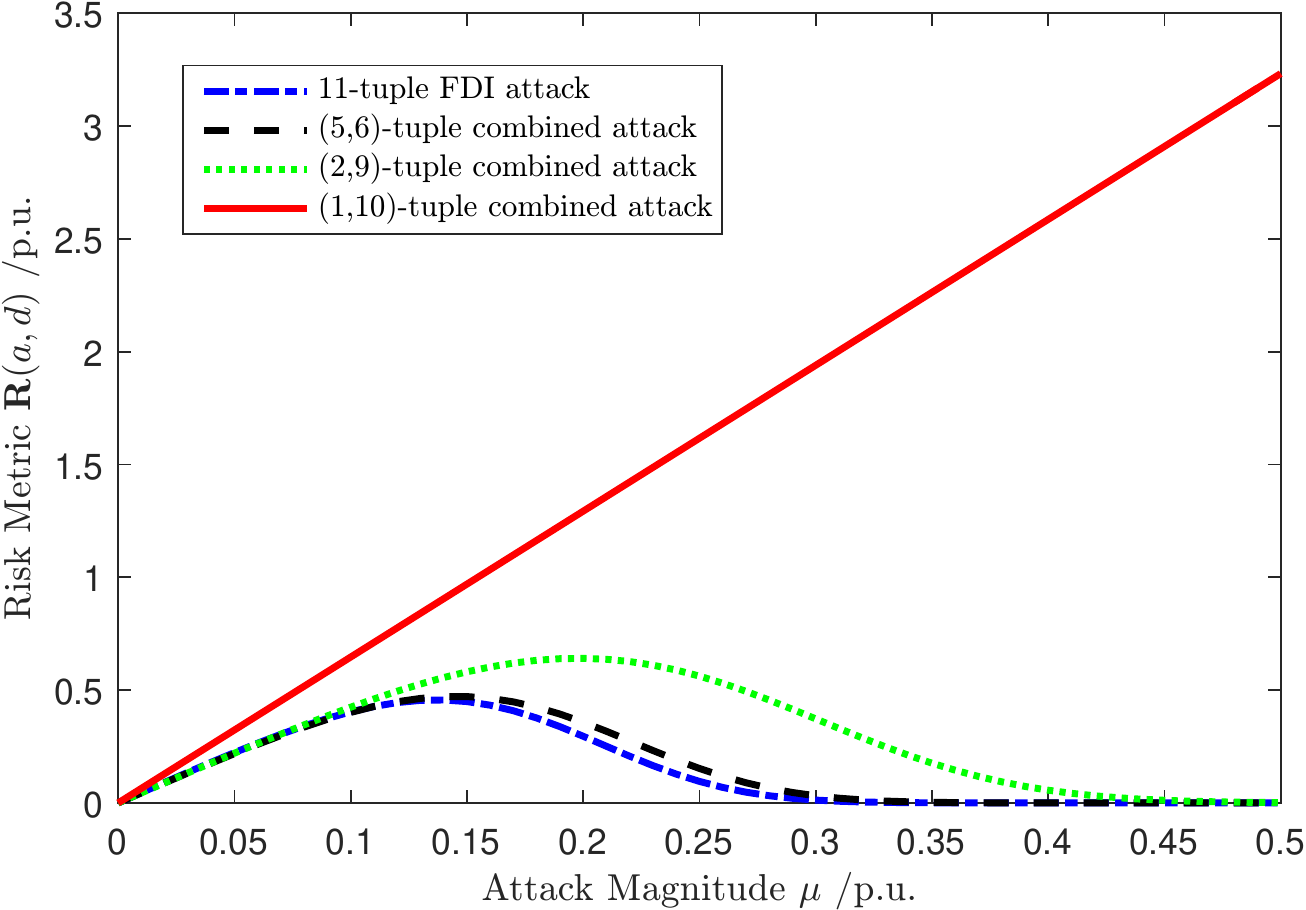}{The risk metric is plotted versus the attack magnitude. The attacks are under structured uncertainty model and performed in the 
set of 11 measurements. 
Here we assume $C_A = C_I$ and 
false alarm rate $\alpha$ is 0.05.} 

\section{Discussion and Conclusion}


In this paper we see that combined attacks can succeed with less resources (if $C_A < C_I$) and lower detection probability when the adversarial knowledge is limited, bringing more risk to reliable system operation. It also should be noted that this paper assumes that the SE treats unavailable measurements due to attacks as a case of missing data, although the amount of missing data under attacks is larger than the one under normal conditions. 
In addition, availability attacks like DoS attacks could trigger alerts on ICT-specific measures (e.g., intrusion detection). 
These two features give the opportunities to develop better cross-domain detection schemes for availability portion of the attacks improving the overall combined attacks detection. 
Other research directions to explore in the future include evaluating physical impact of combined attacks and exploring the vulnerability of other monitoring/control applications to combined attacks.

\ifCLASSOPTIONcaptionsoff
  \newpage
\fi


\bibliographystyle{IEEEtran}
\bibliography{IEEEabrv,Literature_SmartGrid}
\end{document}